\documentclass[12pt]{amsart}
\usepackage{fullpage}

\usepackage{graphicx}

\usepackage{amsfonts}
\usepackage{amssymb}
\usepackage{amsmath}
\usepackage{amsthm}
\usepackage{latexsym}

\newtheorem{theorem}{Theorem}[section]
\newtheorem{lemma}[theorem]{Lemma}

\newtheorem{proposition}[theorem]{Proposition}

\theoremstyle{plain}
\theoremstyle{plain}

\theoremstyle{remark}
    \newtheorem{remark}[theorem]{Remark}

\newcommand{\EE}{\mathbb{E} }
\newcommand{\1}{\mathbf{1} }
\newcommand{\PP}{\mathbb{P} }

\newcommand{\RR}{\mathbb{R} }

\newcommand{\BS}{C_\mathrm{BS}}
\newcommand{\IBS}{Y_\mathrm{BS}}
\renewcommand{\d}{\mathrm{d}}

\begin{document}

\author{Michael R. Tehranchi \\
University of Cambridge }
\address{Statistical Laboratory\\
Centre for Mathematical Sciences\\
Wilberforce Road\\
Cambridge CB3 0WB\\
UK}
\email{m.tehranchi@statslab.cam.ac.uk}

\date{\today}

\title[Bounds for implied volatility]{Uniform bounds for Black--Scholes implied volatility}
\maketitle
 
\begin{abstract}
In this note, Black--Scholes implied volatility is expressed 
in terms of various optimisation problems. From these representations, upper and lower
bounds are derived which hold uniformly across moneyness and call price.
Various symmetries of the Black--Scholes formula are exploited to derive new bounds
from old.
These bounds are used to reprove asymptotic formulae for implied volatility
at extreme strikes and/or maturities. 
\end{abstract}

\section{Introduction}
We define the Black--Scholes call price function $\BS: \RR \times [0,\infty) \to [0,1)$  by the formula
\begin{align*}
\BS(k,y) &=  \int_{-\infty}^{\infty} (e^{yz - y^2/2}- e^k )^+ \phi(z) d z  \\
& = \left\{ \begin{array}{ll}
				 {\Phi}\big(-\frac{k}{y} + \frac{y}{2}\big) - e^k {\Phi}
				\big(-\frac{k}{y} - \frac{y}{2}\big) & \mbox{ if } y > 0 \\
				(1-e^k)^+ & \mbox{ if } y =0,
				\end{array} \right.
\end{align*}
where $\phi(z) = \frac{1}{\sqrt{2\pi}} e^{-z^2/2}$ is the standard normal density and $\Phi(x) = \int_{-\infty}^x \phi(z) d z$ is
its distribution function.   As is well known, the financial significance of the function $\BS$ is that,
within the context of the Black--Scholes model \cite{BS}, 
the minimal replication cost  of a European call option with strike $K$ and maturity $T$  written on a stock 
with initial price $S_0$ 
is given by  
$$
\mbox{replication cost } = S_0 e^{-\delta T} \BS\left[ \log\left( \frac{K e^{-rT}}{S_0 e^{-\delta T} } \right), \sigma \sqrt{T} \right]
$$
where   $\delta$ is the dividend rate, $r$ is the interest rate and $\sigma$
is the volatility of the stock.  Therefore, in the definition of $\BS(k,y)$, the first argument $k$ plays the role of log-moneyness 
of the option and the second
argument $y$ is the total standard deviation of the terminal log stock price.

Of the six parameters appearing in the Black--Scholes formula for the replication cost, five are readily observed in the
market.  Indeed, the strike $K$ and maturity date $T$ are specified by the option contract, and the initial stock
price $S_0$ is quoted.   The interest rate is the yield of a zero-coupon bond $B_{0,T}$ with maturity $T$ and unit face value, 
and can be computed from the initial bond price   $B_{0,T} = e^{-rT}$.   Similarly, the dividend rate can computed
from the stock's initial time-$T$ forward price $F_{0,T} = S_0 e^{ (r-\delta)T}$.   

As suggested by Latan\'e \& Rendleman \cite{LataneRendleman} in 1976,
the remaining parameter, the volatility
$\sigma$, can also be inferred from the market, assuming that the call has a quoted price $C_{0,T,K}$.   
Indeed, note that for fixed $k$, the map $\BS(k, \cdot)$ is strictly
increasing and continuous, so we can define the inverse function 
$$
\IBS(k, \cdot): [ (1-e^k)^+, 1 ) \to [0,\infty)
$$
by 
$$
y = \IBS(k, c) \Leftrightarrow   \BS(k,y) = c.
$$
The  implied volatility of the call option is then defined to be
$$
\sigma^{\mathrm{implied}}  = \frac{1}{\sqrt{T}}
\IBS\left[ \log\left( \frac{K e^{-rT}}{S_0 e^{-\delta T} } \right), \frac{C_{0,T,K}}{S_0 e^{-\delta T}} \right].
$$

Because of its financial significance,   the function $\IBS$  has been the subject of much interest.  
For instance, approximations for $\IBS$ can be found in several papers \cite{BrenSub, corrado-miller,  li, pianca}.
Unfortunately,
there seems to be only one case where the function $\IBS$ can be expressed explicitly in terms of elementary functions:  when $k=0$ we have
\begin{align*}
\BS(0,y ) &= 2 \ \Phi\left( \frac{y}{2} \right) -1 \\
& = 1 - 2 \ \Phi\left( -\frac{y}{2} \right) 
\end{align*}
and hence
\begin{align*}
\IBS(0,c) &= 2\ \Phi^{-1} \left( \frac{1+c}{2} \right)   \\
& = - 2 \ \Phi^{-1} \left( \frac{1-c}{2} \right).
\end{align*}

The main contribution of this article is to provide bounds on the 
quantity $\IBS(k,c)$ in terms of elementary functions of $(k,c)$.
As an example, in Proposition \ref{th:cto1} below we will see that
\begin{equation}\label{eq:firstbound}
 \IBS(k,c ) \le -2  \Phi^{-1} \left( \frac{1-c}{1+e^{ k}} \right)
\end{equation}
for every $(k, c)$ such that $(1-e^k)^+ \le c < 1$.   

We list here two possible applications of such bounds.  
When $k \ne 0$,  the function $\IBS$ can be evaluated numerically. 
A simple way
  to  do so  is to implement the bisection method for finding the root of the map $y \mapsto \BS(k,y)- c$.   
That is to say, for fixed $(k,c)$   pick two points $\ell < u$ such that $\BS(k,\ell) < c$ and $\BS(k, u) > c$.  By the 
intermediate value theorem, we know that the root is in  the the interval $(\ell, u)$. We then let $m=\frac{1}{2}(\ell+u)$
be the midpoint.   
If $\BS(k,m) > c$ we know that the root $\IBS(k,c)$ is in  the the interval $(\ell, m)$, in which case we relabel $m$ as $u$.
  Similarly, if $\BS(k,m) < c$ we relabel $m$ as $\ell$. This process is repeated until 
	$|\BS(k,m)-c| < \varepsilon$, where $\varepsilon > 0$ is a given tolerance level whereupon we
declare $\IBS(k,c) \approx m$.  (We note here that a more sophisticated idea is apply the Newton--Raphson method as
suggested by Manaster \& Koehler \cite{manaster-koehler} in 1982.  We will return to this idea in Section \ref{se:main}.)

In order to implement the bisection method, we need a lower bound $\ell$ and upper bound $u$  to initialise
the algorithm.  
However, aside from the obvious lower bound $\ell = 0$, there do not seem to be many well-known
explicit upper and lower bounds
on the quantity $\IBS(k,c)$ which hold \textit{uniformly} in  $(k,c)$.  This note provides such bounds,
and indeed, equation \eqref{eq:firstbound} is an example.

We now consider another application of our bounds.  Consider a market model
with a zero-coupon bond  with maturity date $T$ whose time-$t$ price is
$B_{t,T}$ and  a stock with time $t$-price $S_t$.  Suppose the initial price
of a call option with strike $K$ and maturity $T$ is given by
$$
C_{0,T,K} = B_{0,T} \EE^{T}[ (S_T- K)^+ ]
$$
where the expectation is under a fixed $T$-forward measure.  Further, suppose the
stock's initial time-$T$ forward price   is given by
$$
F_{0,T} = \EE^{T}[ S_T ].
$$
(If the stock pays no dividend, static arbitrage considerations would imply $F_{0,T} = S_0/B_{0,T}$.  
We do not need this formula here so the stock is allowed to pay dividends in the present discussion; 
however, we will
return to this point in Remark \ref{re:bubble} below.)  Now, 
equation \eqref{eq:firstbound} implies that the implied volatility is bounded by
\begin{align}
\sigma^{\mathrm{implied}} & =  
\frac{1}{\sqrt{T}} \IBS\left[ \log \left(\frac{K}{F_{0,T}}\right), \frac{C_{0,T,K}}{F_{0,T} B_{0,T}} \right]
\label{eq:intro} \\
& \le - \frac{2}{\sqrt{T}} \Phi^{-1} \left( \frac{\EE^T [ S_T \wedge K ]}{\EE^T[ S_T ] + K } \right). \nonumber
\end{align}
Note that the above bound is the composition of two ingredients: the model-\textit{dependent}
formulae for the quantities $C_{0,T,K}$ and $F_{0,T}$, and a uniform and
model-\textit{independent} bound on the function $\IBS$.

There has been much recent interest in implied volatility asymptotics.  See
for instance the papers \cite{
BenaimFriz, BenaimFriz2, CC, DHJ, GaoLee, gulisashvili, gulisashvili2, gulisashvili3, lee, RR, JAP09} for 
asymptotic formulae   which depend on
minimal model data, such as the distribution function or the moment generating function
of the returns of the underlying stock.  Paralleling the discussion above, such asymptotic formulae can
be seen as compositions of two limits:  first, the asymptotic shape of the call surface as predicted by the model at, for 
instance, extreme strikes
and/or maturities; and second, asymptotics of the model-independent function $\IBS$.  
The uniform bounds on $\IBS$ that are presented in this note are used to provide short, new
proofs of these second model-independent asymptotic formulae.

 In their long survey article, 
 Andersen \& Lipton \cite{AndersenLipton} warn that   many of the
 asymptotic implied volatility formulae that have appeared in 
recent years may not be applicable in practice, since typical market parameters
are usually not in the range of validity of any of the proposed asymptotic regimes.  
Our new bounds on the function $\IBS$ are uniform, and hence side-step the
critique of Andersen \& Lipton.

The rest of the note is organised as follows.  
In Section \ref{se:sym} we discuss various symmetries of the Black--Scholes
pricing function $\BS$.  These symmetries will be used repeatedly throughout
the remainder of the note.  In Section \ref{se:main} the Black--Scholes
implied total standard deviation function $\IBS$ is represented as the value
function of several optimisation problems.  These results constitute
the main contribution of this note  since they allow $\IBS$ to be bounded arbitrarily
well from above and below by choosing suitable controls to insert into the respective
objective functions.  
In Section \ref{se:asym} these bounds are used to reprove some known asymptotic
formulae.  As a by-product, we derive formulae which have the same asymptotic
behaviour as the known formulae, but are guaranteed to bound $\IBS$ either from
above or below.   

\section{Put-call and close-far symmetries}\label{se:sym}
The Black--Scholes call price function $\BS$ contains a certain amount of symmetry.
In order to streamline the presentation of our bounds, we begin with
an exploration of two of these symmetries.

To treat the two cases $k \ge 0$ and $k < 0$ as 
efficiently as possible, we  begin with an observation.
Suppose $c$ is the normalised price of a call option with log-moneyness $k$.  Then
by the usual put-call parity formula, the corresponding normalised price of a put option with the same log-moneyness is
$$
p = c + e^k - 1.
$$ 
Now if $c = \BS( k, y)$ is for some $y > 0$, then we have
\begin{align*}
p &=  \BS( k, y) + e^k - 1 \\
 &= e^k \Phi\left(\frac{k}{y}+ \frac{y}{2} \right) - \Phi\left(\frac{k}{y}- \frac{y}{2} \right) \\
& = e^k \BS( - k, y).
\end{align*}
The above calculation is the well-known Black--Scholes put-call symmetry formula.   
We have just proven the following result:
\begin{proposition}\label{th:sym}
For any $k \in \RR$ and  $c \in [(1-e^k)^+, 1)$ we have 
$$
\IBS(k, c) = \IBS(-k, e^{-k}c + 1 - e^{-k} ).
$$
\end{proposition}

One conclusion of proposition \ref{th:sym} is that it is sufficient to study the function $\IBS(k, \cdot)$
only in the case $k \ge 0$. Indeed, to study the case $k< 0$ one simply applies the above put-call symmetry
formula. 

We now come to another, less well-known, symmetry of the Black--Scholes formula.
While put-call symmetry involves replacing the log-moneyness $k$ with $-k$, the symmetry discussed here involves
replacing the total standard deviation $y$ with $2|k|/y$.    
By put-call symmetry, we can confine our discussion to the case $k > 0$.  

\begin{proposition}\label{th:inv} For all $k > 0$ and $0 < c < 1$, let
$$
\hat C(k, c) = 1 - \int_0^c \frac{ 2k}{ [ \IBS(k, u) ]^2 } du.
$$
Then $\hat C(k, c)  > 0$ and we have
$$
\IBS(k,c) = \frac{2k}{\IBS(k, \hat C(k, c) )}.
$$
\end{proposition}

Figure \ref{fi:c-hat} is shows the graph of $c \mapsto \hat C(k, c)$ when $k=0.2$.

\begin{figure}  	
\caption{The function $\hat C(k, \cdot)$.}
	  \label{fi:c-hat}
		\includegraphics[trim = 0cm 0 0cm 0, clip, scale = 0.45]{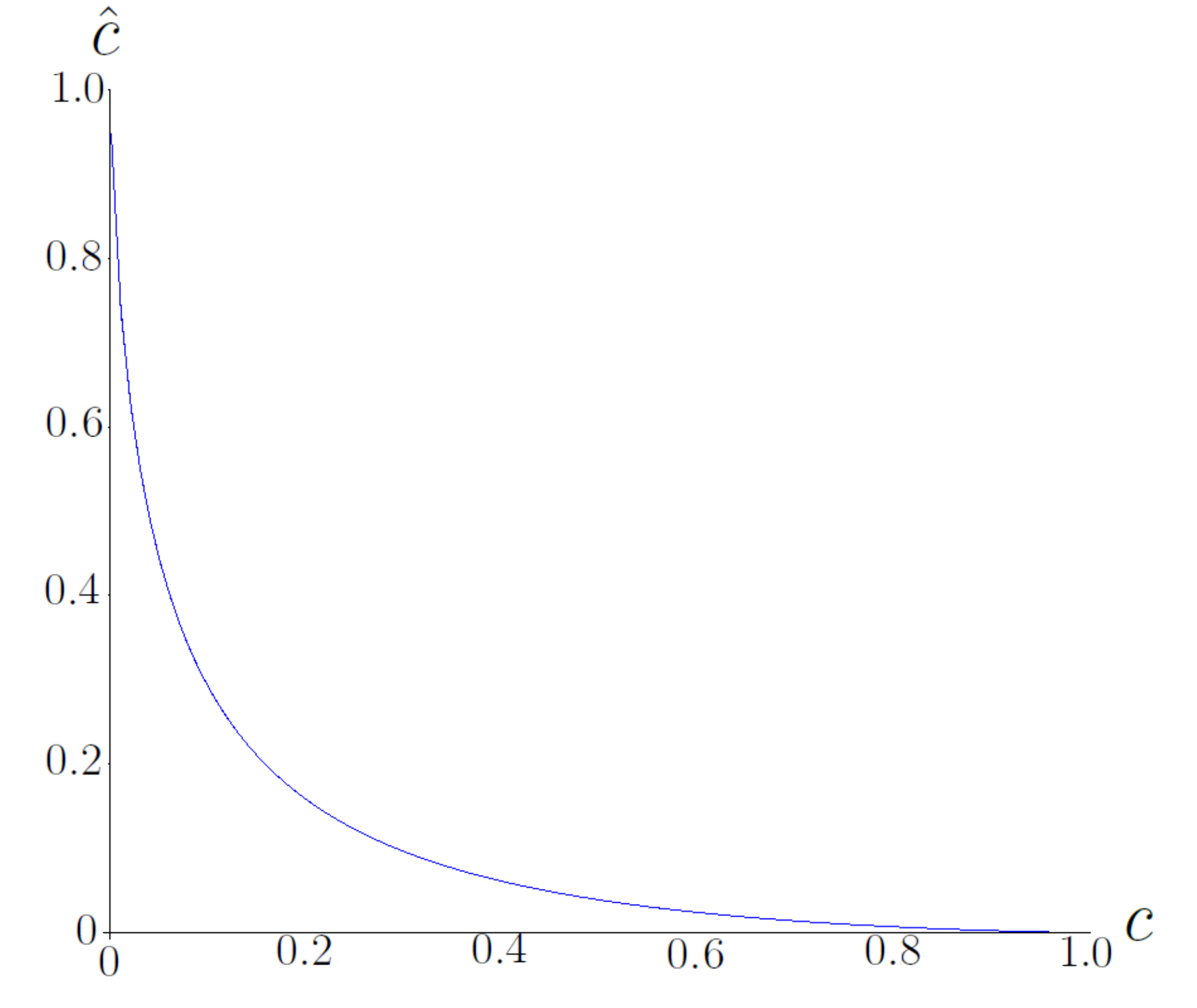}		
	\end{figure} 
	
\begin{proof} We must prove that
$$
\hat C \left(k,  c   \right) = \BS\left(k, \frac{2k}{\IBS(k,c)} \right),
$$
or equivalently
$$
\hat C \left(k,  \BS(k,y)   \right) = \BS\left(k, \frac{2k}{y} \right),
$$
The above identity can be verified by differentiating both sides with respect to $y$, 
and using the Black--Scholes vega formula: for $k > 0$, we have
$$
 \BS(k,y) = \int_0^y \phi( -k/x + x/2) dx.
$$ 
\end{proof}

\begin{remark}
Fix $k> 0$ and $y > 0$, and let $c = \BS(k,y)$.   Note that $c \approx 0$ when $y$ is very small, 
and indeed it is a straightforward exercise to verify (see Section \ref{se:asym}) that 
$$
\log c = - \frac{k^2}{2 y^2} + O(\log y) \mbox{ as } y \downarrow 0.
$$
 On the other hand, we have $c \approx 1$ when $y$ is very large,
and furthermore
$$
\log(1-c) = - \frac{y^2}{8}    + O(\log y) \mbox{ as } y \uparrow \infty.
$$

Now, let $\hat c = \BS(k, 2k/y)$ so that by Proposition \ref{th:inv}
we have $\hat c = \hat C(k,c)$.
From the above calculations we have
$$
\log (1- \hat c) = - \frac{k^2}{2 y^2} + O(\log y) \mbox{ as } y \downarrow 0,
$$
and
$$
\log \hat c  = - \frac{y^2}{8}   + O(\log y)  \mbox{ as } y \uparrow \infty.
$$

In the context of the Black--Scholes model, the quantity $y$
has the interpretation as the total
standard deviation $y= \sigma \sqrt{T}$, where $\sigma$
is the volatility and $T$ is the maturity date of the option.  
Proposition \ref{th:inv} then is a symmetry  relation between the prices of 
short-dated and long-dated options. 
\end{remark}

We conclude this section with some easy observations which we will use later.

\begin{proposition}\label{th:inv2}  For all $k > 0$,
the function $\hat C(k, \cdot)$ is convex and satisfies
the functional equation 
$$
\hat C(k, \hat C(k, c) ) = c
$$
holds for all $0 < c < 1$. 
\end{proposition}

\begin{proof}
It is easy to see that $\IBS(k, \cdot)$ is strictly increasing.
That $\hat C(k, \cdot)$ is convex follows from the fact that its
gradient $-2k/\IBS(k, \cdot)^2$ is increasing.  

That the functional equation is proven by noting
$$
\IBS(k, c) = \frac{2k}{\IBS(k, \hat C(k,c)) } = \IBS(k, \hat C(k,\hat C(k,c)))
$$
and using the fact that that $\IBS(k, \cdot)$ is strictly increasing.
\end{proof}

\begin{proposition}\label{th:J-funct}
For $k > 0$, let
$$
J(k,c) = \int_0^c \frac{1}{\IBS(k,u)} du.
$$
Then
$$
J(k,c) + J(k, \hat{c} ) = J(k,1) 
$$
where $\hat c = \hat C(k,c)$.
\end{proposition}

\begin{proof}
By setting $c = \BS(k,y)$ and hence $\hat c= \BS(k, 2k/y)$, the identity can be 
proven by computing the derivative with respect to $y$ of the left-hand side, and note
that it is vanishes identically.
\end{proof}

\begin{remark}
By changing variables, we have the identities
\begin{align*}
J(k,1) & = \int_{0}^{\infty} \frac{\phi(-k/y+y/2)}{y} dy \\
& = \int_{-\infty}^{\infty} \frac{ \phi(x) }{\sqrt{ x^2 + 2k}} dx \\
& = \frac{e^{k/2}}{\sqrt{2\pi}} K_0 (k/2)
\end{align*}
where $K_0$ is a modified Bessel function.  See \cite{table}.
\end{remark}

\section{Various optimisation problems}\label{se:main}
This section contains one of the main results of this note, formulae for
the function $\IBS$ in terms of various optimisation problems.  
The first result  is that $\IBS(k,c)$ can be calculated by solving
a minimisation problem.  In particular, we can use this formula to find an upper
bound simply by evaluating the objective function at a feasible control.

\begin{theorem}\label{th:main} 
For all $k \in \RR$ and $(1-e^k)^+ \le c < 1$ we have
\begin{align*}
\IBS(k,c) &= \inf_{d_1 \in \RR} [d_1-\Phi^{-1} \big( e^{-k}(\Phi(d_1) -c) \big) ] \\
& = \inf_{d_2 \in \RR}  [ \Phi^{-1}\big( c + e^k \Phi(d_2) \big) - d_2  ] \\
\end{align*}
Furthermore, if $c > (1-e^k)^+$, then the two infima are attained at
\begin{align*}
d_1^* &= -\frac{k}{y} + \frac{y}{2}, \\
d_2^* &= -\frac{k}{y} - \frac{y}{2}
\end{align*}
where $y = \IBS(k,c)$.
\end{theorem}

\begin{remark} We are using the convention that $\Phi^{-1}(u) = + \infty$ for $u \ge 1$
and $\Phi^{-1}(u) = -\infty$ for $u \le 0$.
\end{remark}

 The following proof 
is due to Pieter-Jan De Smet \cite{S}, simplifying the proof in an early version of this paper.
The idea is essentially that the inequality 
$$
(X-K)^+ \ge (X-K) \1_{[H, \infty)}(X)
$$
 holds for all
$X,K, H \ge 0$ with equality if and only if $H=K$.
\begin{proof}
 Fix $k \in \RR$ and $(1-e^k)^+ \le c < 1$ and let
 $y \ge 0$ be such that $\BS(k,y) =c$. 
Note that for any $d_2 \in \RR$ we have
\begin{align*}
c &= \int_{-\infty}^{\infty} (e^{y z - y^2/2} - e^k)^+ \phi(z) \d z \\
& \ge \int_{-d_2}^{\infty} (e^{y z - y^2/2} - e^k)^+  \phi(z) \d z \\
& \ge \int_{-d_2}^{\infty} (e^{y z - y^2/2} - e^k)  \phi(z) \d z \\
& = \Phi(d_2+y) - e^k \Phi(d_2).
\end{align*}
There is equality from the first to second line only if $-d_2  \le  k/y + y/2$,
and there is equality from the second to the third line only if $- d_2 \ge  k/y+ y/2$.
Rearranging then yields
$$
y \le \Phi^{-1}(c + e^k \Phi(d_2)) - d_2.
$$
Let $d_2 =  \Phi^{-1} \left( e^{-k}(\Phi(d_1) -c) \right) $ in the above inequality to 
obtain the first expression.
	\end{proof}
		
Let 
\begin{equation}\label{eq:H1}
H_1(d; k, c) = d-\Phi^{-1} \big( e^{-k}(\Phi(d) -c) \big) 
\end{equation}
and
\begin{equation}
H_2(d; k, c) =  \Phi^{-1}\big( c + e^k \Phi(d) \big) - d,
\end{equation}
and note that
$$
H_1(d; k, c) = H_2(-d; -k, e^{-k} c + 1 - e^{-k} )
$$
in line with put-call symmetry.  We use this notation
to compute $\IBS(k,c)$ in terms of a maximisation problem.  This representation can be used, in principle,
to find lower bounds.

\begin{theorem}\label{th:lowerbound}
Let $\mathcal C$ be the space of continuous functions on $[0,1]$.  
For $k > 0$ and $0 < c < 1$,  we have
$$
\IBS(k,c) = \sup_{D  \in \mathcal C, d  \in \RR} \frac{2k}{H_i \left( d ; k, 1 - \int_0^c \frac{2k}{  H_j(D(u); k, u)^2 } du \right)}
$$
for any $i,j \in \{1,2\}$.
\end{theorem}

\begin{proof}
By Theorem \ref{th:main} we have $ \IBS(k, u) \le H_j(d; k , u )$ for all $d$, and 
since $\IBS(k, \cdot)$ is increasing, we  have for any $ D \in \mathcal C$ that
\begin{align*}
\IBS( k, \hat C(k, c) ) &  = \IBS \left( k, 1 - \int_0^c \frac{2k}{ \IBS(k, u)^2} du \right) \\
& \le \IBS\left( k,  1 - \int_0^c \frac{2k}{  H_j(D(u); k, u)^2}du \right)  \\
& \le H_i\left( d; k, 1 - \int_0^c \frac{2k}{  H_j(D(u); k, u)^2} du \right).
\end{align*}
The conclusion follows from Proposition \ref{th:inv}.
\end{proof}

In light of Proposition \ref{th:inv} we now give a representation
of $\hat C$ in terms of a minimisation problem.
We restrict attention to $k > 0$ with no real loss thanks to put-call symmetry.

\begin{proposition}
For $k > 0$ and $0 < c < 1$ we have
$$
\hat C(k,c) = \sup_{y \ge 0} \left[ \BS\left(k, \frac{2k}{y} \right) - \frac{2k}{y^2}( c - \BS(k,y)) \right]
$$
\end{proposition}

\begin{proof}
Recall that by Proposition \ref{th:inv2} that $\hat C(k, \cdot)$ is convex.  Hence
$$
\hat C(k, c) - \hat C(k, c^*) \ge -\frac{2k}{ \IBS(k, c^*)^2} ( c - c^*).
$$
for any $c, c^* \in (0,1)$.  Letting $y = \IBS(k, c^*)$ we have 
$$
\hat C(k, c) \ge  \BS\left(k, \frac{2k}{y} \right)  -\frac{2k}{y^2} ( c - \BS(k,y))
$$
as claimed.
\end{proof}

Of course, there are other representations of $\IBS$ in terms of an optimisation problems.
For instance, 
we have
\begin{align*}
\IBS(k,c) &= \inf\{ y  \ge 0:  \BS(k,y) \ge c \} \\
& = \sup\{ y \ge 0:  \BS(k,y) \le c \}.
\end{align*}
Indeed, this simple observation underlies the bisection method discussed in the introduction.

We conclude this section with a slightly more interesting representation. It   be can used to
find upper and lower bounds of $\IBS(k,c)$, at least in principle.  However, 
in practice it is not clear how to choose candidate controls, so 
we do not explore this idea in the sequel.  This result is due to
Manaster \& Koehler \cite{manaster-koehler}, and is motivated by the Newton--Raphson method for
computing implied volatility numerically.

\begin{proposition} [Manaster \& Koehler] Fix $k \ge 0$ and $0 \le c < 1$.  
If $c \le \BS(k, \sqrt{2k}) = 1/2 - e^k \Phi(-  \sqrt{2k} )$ then
$$
\IBS(k, c) = \inf_{0 \le y \le \sqrt{2k}} \left[ y + \frac{c - \BS(k, y)}{\phi( -k/y+y/2)} \right].
$$
Otherwise, if $c \ge    1/2 - e^k \Phi(- \sqrt{2k} )$ then
$$
\IBS(k, c) = \sup_{  y \ge \sqrt{2k}}\left[ y + \frac{c - \BS(k, y)}{\phi( -k/y+y/2)} \right].
$$
\end{proposition}

\begin{proof}  The restriction of 
$\BS(k, \cdot)$ to $[0, \sqrt{2k} ]$ is convex, as can be confirmed by differentiation.  Hence, by the Black--Scholes
vega formula, we have 
$$
\BS(k, y^*) - \BS(k, y) \ge \phi( - k/y + k/y)(y^* - y)
$$
for any $y, y^* \in [0, \sqrt{2k} ]$.  Fixing $y^*$ and letting $c= \BS(k, y^*)$  we have
proven
$$
y^* \le  y + \frac{c - \BS(k, y)}{\phi( -k/y+y/2)}
$$
as desired.  
 Similarly, since the restriction of $\BS(k, \cdot)$ to $[ \sqrt{2k}, \infty)$ is concave the 
second conclusion follows.
\end{proof}

\section{Uniform bounds and asymptotics}\label{se:asym}
In this section, we will offer quick proofs of some asymptotic formulae for the function $\IBS$.  
These formulae already appear in the literature, but the important novelty here is that
we will derive bounds on the function $\IBS$ which hold uniformly, not just asymptotically.
To obtain upper bounds in most cases, we simply choose a convenient $d_1$ or $d_2$ to plug into Theorem \ref{th:main}.
 Note that the proposed upper bound is close to the true value of $\IBS(c,k)$ when, for instance, the proposed value 
of $d_1$ is close to the minimiser $d_1^* = - k/y +y/2$.
In principle, lower bounds could be found by choosing convenient controls into 
Theorem \ref{th:lowerbound}.  However, in practice, we have found other arguments, while
   lacking the same   unifying principle, which do have the advantage of being simple.
In the proofs that follow, we  usually only consider the $k \ge 0$ case, as the $k<0$
case follows directly from Proposition \ref{th:sym}. 

Before we begin, 
 we need a lemma regarding the asymptotic behaviour of the standard normal quantile function $\Phi^{-1}$.
\begin{lemma}\label{th:asym} As $\varepsilon \downarrow 0$ we have
$$
\left[\Phi^{-1}(\varepsilon)\right]^2 =  -2 \log \varepsilon + O( \log( - \log \varepsilon) ).
$$
 In particular, we have
$$
\Phi^{-1}(\varepsilon) = - \sqrt{ -2 \log \varepsilon} + O\left( \frac{\log( - \log \varepsilon)}{ \sqrt{-\log \varepsilon} }   \right)
$$
\end{lemma}

\begin{proof}
Let $\varepsilon = \Phi(-x)$ for  large $x >0$ and let
$$
R(x) = \frac{ \Phi(-x) x}{\phi(x)}.
$$ 
In this notation we have the identity
$$ 
\log \Phi(-x) = - x^2/2   -  \log( \sqrt{2\pi} x) +  \log R(x).
$$
Since it is well known that $R(x) \to 1$ as $x \to \infty$ we have
$$
\frac{\log \Phi(-x)}{x^2} \to -1/2   
$$
or equivalently
$$
\frac{[\Phi^{-1}(\varepsilon)]^2}{\log \varepsilon} \to -2.
$$
Plugging in this limit into the identity yields the first conclusion, and Taylor's theorem yields the second.
\end{proof}

The first example comes from \cite{JAP09}.  
This asymptotic formula considers the behaviour of $\IBS$ when $c$ is close to its upper bound of $1$. 
This result is useful in studying implied
volatility at very long maturities, when the strike is fixed.
\begin{theorem}\label{th:jap} For fixed $k \in \RR$, we have
$$
\IBS(k,c) = \sqrt{- 8  \log (1-c)} + O \left( \frac{ \log[- \log (1-c) ]}{\sqrt{ - \log(1-c)} } \right)
$$
as $c \uparrow 1$. 
\end{theorem}

The proof of the above theorem relies the following simple bounds which hold uniformly in $(c,k)$.

\begin{proposition} \label{th:cto1}  Fix $k \in \RR$ and $(1-e^k)^+ \le c < 1$.
For $k \ge 0$ we have
$$
 - 2 \Phi^{-1}\left(\frac{1-c }{2   }\right) \le   \IBS(k,c ) \le -2  \Phi^{-1} \left( \frac{1-c}{1+e^{ k}} \right)
$$
and for $k < 0$ we have
$$
 - 2 \Phi^{-1}\left(\frac{1-c }{2e^k   }\right) \le   \IBS(k,c ) \le -2  \Phi^{-1} \left( \frac{1-c}{1+e^{ k}} \right).
$$
\end{proposition}

\begin{proof}  For the upper bound, let $d_2 = \Phi^{-1} \left( \frac{1-c}{1+e^k} \right) $ in
 Theorem \ref{th:main}.

For the lower bound, let  $y = \IBS(k,c)$.  Note that
$\BS( \cdot, y)$ is decreasing and hence
\begin{align*}
1 - 2 \Phi( -y/2) & = \BS( 0, y ) \\
& \ge \BS(k,  y) \\
& =   c
\end{align*}
when $k \ge 0$.   In the case when $k < 0$, note that 
$$
\frac{1-e^{-k} p}{1+e^{-k}} = \frac{1-c }{1+  e^k}
$$
and that 
$$
\frac{1-e^{-k} p}{2} = \frac{1-c }{2  e^k}.
$$
Now appeal to the put-call parity formula of Proposition \ref{th:sym}.
 \end{proof}

\begin{proof} [Proof of Theorem \ref{th:jap}] By Proposition \ref{th:cto1} and  Lemma \ref{th:asym}, we have
\begin{align*}
\IBS(k,c ) &\le -2  \Phi^{-1} \left( \frac{1-c}{1+e^k} \right) \\
& = \sqrt{- 8  \log (1-c)} + O \left( \frac{ \log[- \log (1-c) ]}{\sqrt{ - \log(1-c)} } \right)
\end{align*}
where we have used the fact that for fixed $k$ we have
$$
\sqrt{ -2 \log\left(\frac{1-c}{1+e^k}\right)} =  \sqrt{ -2 \log(1-c)} + O( \frac{1}{\sqrt{-\log(1- c)}} )
$$
as $c \uparrow 1$.

Similarly, by Proposition \ref{th:cto1}, we have for $k \ge 0$ that
\begin{align*}
 \IBS(k,c ) &\ge -2  \Phi^{-1} \left( \frac{1-c}{2} \right) \\
& = \sqrt{- 8  \log (1-c)} + O \left( \frac{ \log[- \log (1-c) ]}{\sqrt{ - \log(1-c)} } \right).
\end{align*}
The $k <0$ is identical.
\end{proof}

Figure \ref{fi:long} illustrates the behaviour of $\IBS(k,c)$ as $c \uparrow 1$, compared
with the uniform upper and lower bounds of Proposition \ref{th:cto1} and the asymptotic
formula in Theorem \ref{th:jap}.   We fixed the log-moneyness $k=0.2$ and plotted four
functions: 
\begin{enumerate}
\item
 $Y_{\mathrm{upper}}(c) =  -2 \Phi^{-1}\left(\frac{1-c}{1+e^k}\right)$ is the upper bound from Proposition \ref{th:cto1};
\item
  $Y_*(c) = \IBS(k,c)$ is the true function of our interest; 
	\item 
	$Y_{\mathrm{lower}}(c) =   -2 \Phi^{-1}\left(\frac{1-c}{2}\right)$ is the lower bound from Proposition \ref{th:cto1}; 
	\item  $Y_{\mathrm{asym}}(c) =  \sqrt{-8 \log(1-c)}$ is the asymptotic shape from Theorem \ref{th:jap}. 
	\end{enumerate}
	Note that $Y_{\mathrm{upper}} \ge Y_* \ge Y_{\mathrm{lower}}$ as predicted.  Also,
	it is interesting to see that $Y_{\mathrm{lower}}$ is remarkably good 
	approximation over large range of $c$.  Finally,  note that $Y_{\mathrm{asym}} \ge Y_{\mathrm{upper}}$ for this range of $c$.
	Indeed, $Y_{\mathrm{asym}} $ is a rather poor approximation of $Y_*$ for realistic values
	of the normalised call price $c$ due to the fact that the error
	term $\log(-\log(1-c))/\sqrt{-\log(1-c)}$ is actually \textit{increasing} for $c < 1-e^{-e^2} = 0.9994$!
	 
\begin{figure}
\caption{Bounds and asymptotics of $\IBS(k, \cdot)$ as $c \uparrow 1$.}
	  \label{fi:long}
		\includegraphics[trim = 0cm 0 0cm 0, clip, scale = 0.38]{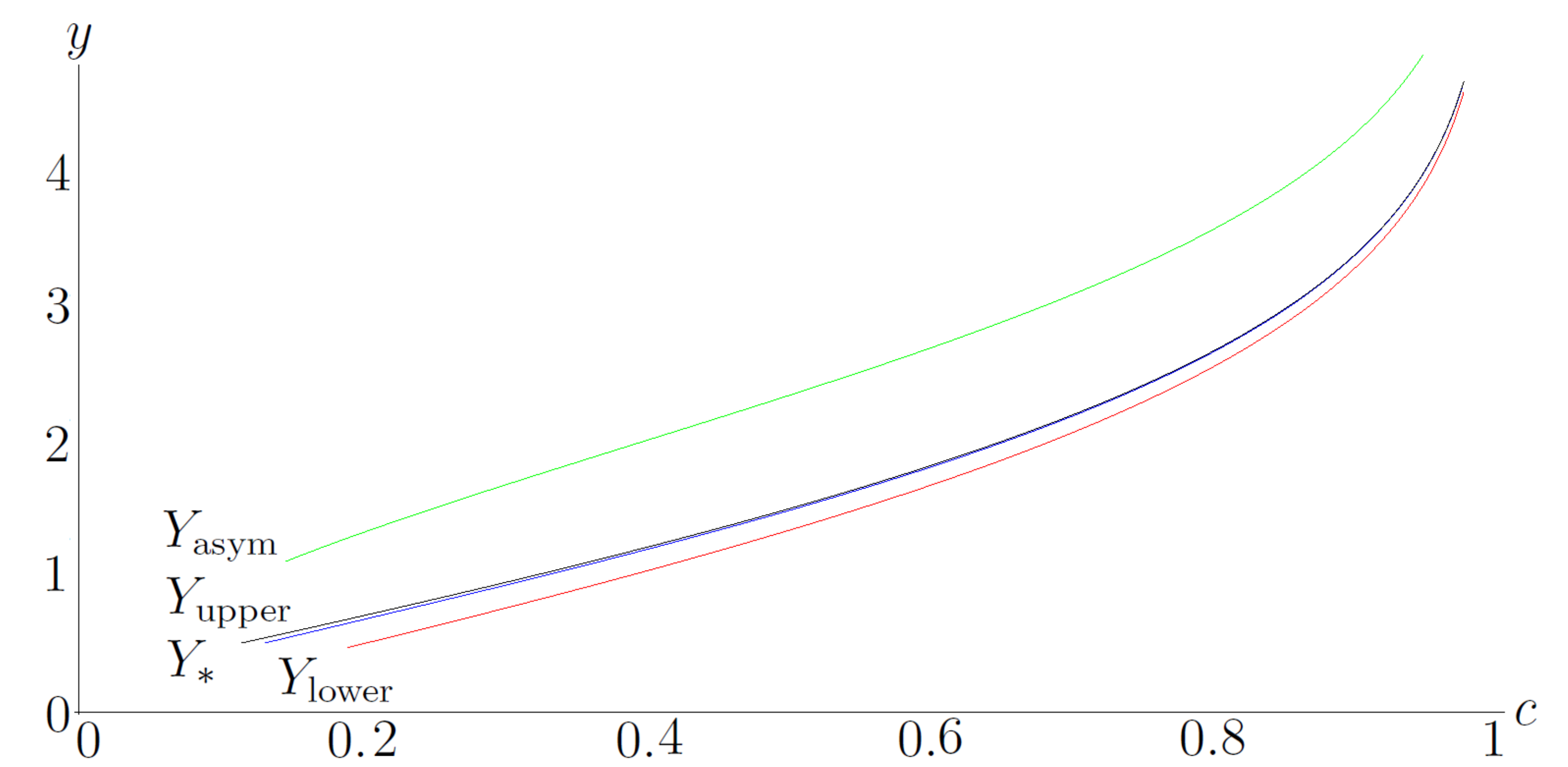}		
	\end{figure} 

The next example we consider in this section is due to Roper \& Rutkowski \cite{RR}
and deals with the case where  $c$ is close 
to its lower bound $(1-e^k)^+$.  In particular, this regime is useful for studying 
the implied volatility smile of options very near maturity.

\begin{theorem}[Roper \& Rutkowski] \label{th:RR}If $k > 0$ then
$$
\IBS(k,c) = \frac{ k}{\sqrt{ - 2 \log c}} + O\left( \frac{ \log(-\log c)}{(-\log c)^{3/2}} \right)
$$
as $c \downarrow 0$.  If $k < 0$ then
$$
\IBS(k,c) = \frac{ -k}{\sqrt{ - 2 \log p}} + O\left( \frac{ \log(-\log p)}{(-\log p)^{3/2}} \right)
$$
as $c \downarrow 1-e^k$,  
where $p = c + e^k - 1$.
\end{theorem}

As always, we will prove the asymptotic result by finding uniform  bounds.
As discussed in Section \ref{se:sym}, we can reuse of the bounds which
are tight when $c$ is close to 1 by first bounding the function $\hat C$.

\begin{proposition}\label{th:hatbound}  For $k >0$ and $0 < c < 1$, we have
$$
1 - c \ L(k,c)  \le \hat C(k, c) \le  1-c
$$
where 
\begin{equation}\label{eq:epsilon}
L(k,c) = \frac{2}{k} \left[\Phi^{-1}\left( \frac{c}{1+e^k}\right)^2+2 \right].
\end{equation}
\end{proposition}

\begin{proof}
For the upper bound, simply note that 
$$
\BS(k,y) + \BS(k,2k/y)= 1 - 2 e^k \Phi(-k/y - y/2) \le 1.
$$

Now
\begin{align*}
\hat C(k, c) & = 1 - \int_0^c \frac{2k}{\IBS(k,u)^2} du \\
& = 1 - \int_0^c \frac{\IBS(k, \hat C(k,u) )^2}{2k} du \\
& \ge 1 - \int_0^c \frac{\IBS(k, 1-u )^2}{2k} du 
\end{align*}
by two applications of Proposition \ref{th:inv} and the upper bound.

Now, we appeal to the upper bound in Proposition \ref{th:cto1} to conclude that
\begin{align*}
\hat C(k, c) &\ge 1 -  \frac{2}{k}\int_0^c \Phi^{-1}\left(\frac{u}{1+e^k}\right)^2 du \\
& = 1  - \frac{2(1+e^k)}{k} \int_{-\Phi^{-1}\left(\frac{c}{1+e^k}\right)}^{\infty} x^2 \phi(x) dx.
\end{align*}
To complete the proof, note that the bound
\begin{align*}
\int_{A}^{\infty} x^2 \phi(x) dx & = A \phi(A) + \Phi(-A) \\
& \le  (A^2 + 2) \Phi(-A)
\end{align*}
which holds for all $A \ge 0$.
\end{proof}

We now prove an inequality which provides
an easy way to convert bounds which are good when $c \uparrow 1$ into
bounds which are good when $c \downarrow 0$.

\begin{proposition}\label{th:short1}  Fix $k > 0$ and $0 <  c < 1$. Then
$$
\frac{2k}{\IBS(k, 1-c)} \le \IBS(k,c) \le 
\frac{2k}{\IBS(k, 1- c \ L(k,c) )}.
$$
where $L(k,c)$ is defined by equation \eqref{eq:epsilon}.  In
particular,  we have
$$
\IBS(k,c) \ge \frac{k}{- \Phi^{-1}\left( \frac{c}{1+e^k} \right)} 
$$
and if $c \ L(k,c) \le 1$ we 
have
$$
\IBS(k,c) \le \frac{k}{- \Phi^{-1} \left( \frac{c \ L(k,c) }{2} \right)}.
$$
\end{proposition}

\begin{proof} The first claim follows from the fact that $\IBS(k, \cdot)$ is 
increasing and Proposition \ref{th:inv}.  The second set of claims
follow from the bounds in Proposition \ref{th:cto1}.
\end{proof}

\begin{remark}
The inequality
$$
\IBS(k, c) \IBS(k,1-c) \ge 2k
$$
which holds for all $k > 0$ and $0 < c < 1$, has an appealing symmetry!
\end{remark}

\begin{proof}[Proof of Theorem \ref{th:RR}]
First fix $k > 0$.  
Using the second lower bound from Proposition \ref{th:short1}, together with Lemma \ref{th:asym}, we have
 that
$$
\IBS(k,c) \ge \frac{k}{-2 \log c} +  O\left( \frac{ \log(-\log c)}{(-\log c)^{3/2}} \right).
$$
Similarly, since  Lemma \ref{th:asym} implies that the quantity $L(k,c)$
from Proposition \ref{th:hatbound} is of asymptotic order
$$
L(k,c)  = O( \log c )
$$
as $c \downarrow 0$ thanks to Proposition \ref{th:asym}, the upper bound follows.  

The case $k< 0$ follows from the put-call symmetry of 
Proposition \ref{th:sym}.
\end{proof}

Figure \ref{fi:short} illustrates the behaviour of $\IBS(k,c)$ as $c \downarrow 0$, compared
with the uniform upper and lower bounds of Proposition \ref{th:short1} and the asymptotic
formula in Theorem \ref{th:RR}.    We fixed the log-moneyness $k=0.2$ and plotted four
functions: 
\begin{enumerate}
\item $Y_{\mathrm{upper}}(c) =  \tfrac{k}{- \Phi^{-1} \left( \frac{c \ L(k,c) }{2} \right)}$ is the upper bound from Proposition \ref{th:short1};
\item   $Y_*(c) = \IBS(k,c)$ is the true function of our interest; 
\item $Y_{\mathrm{lower}}(c) =   \tfrac{k}{- \Phi^{-1}\left( \frac{c}{1+e^k} \right)}$ is the lower bound from Proposition \ref{th:short1}; 
\item  $Y_{\mathrm{asym}}(c) =  \tfrac{ k}{\sqrt{ - 2 \log c}}$ is the asymptotic shape from Theorem \ref{th:RR}.
\end{enumerate}
	Note again that $Y_{\mathrm{upper}} \ge Y_* \ge Y_{\mathrm{lower}}$ as predicted.   
	Finally,  note that $Y_{\mathrm{asym}} \le Y_{\mathrm{lower}}$ for this range of $c$.
	  
\begin{figure}
\caption{Bounds and asymptotics of $\IBS(k, \cdot)$ as $c \downarrow 0$. }
	  \label{fi:short}
		\includegraphics[trim = 0cm 0 0cm 0, clip, scale = 0.38]{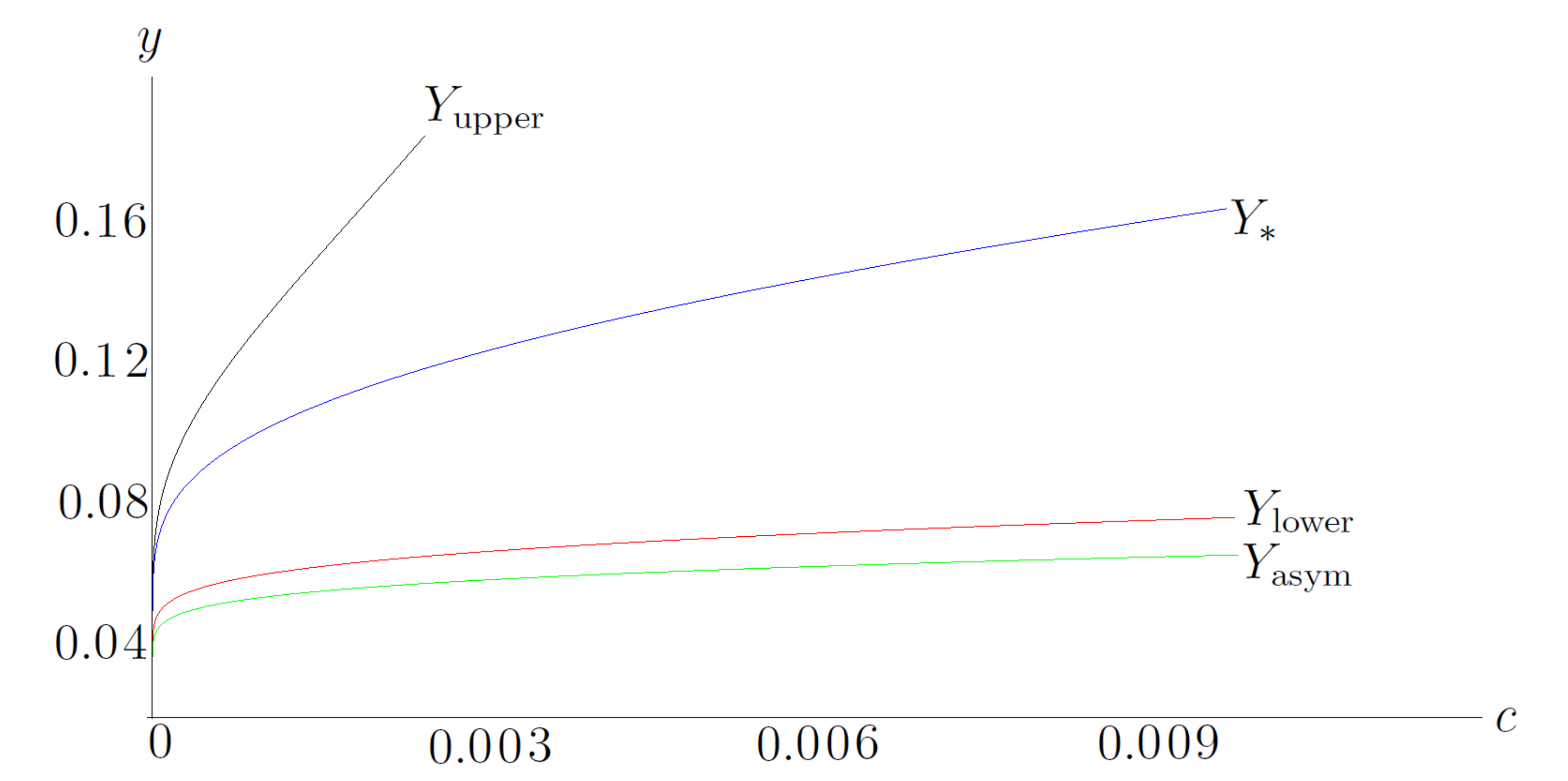}		
	\end{figure} 

The next example is due to Gulisashvili \cite{gulisashvili}.  This result is useful in studying
the wings of the implied volatility surface for extreme strikes but fixed maturity date.
\begin{theorem}[Gulisashvili] \label{th:gul} If $c(k) \downarrow 0$ as $k \uparrow + \infty$ then
$$
 \IBS(k,c(k) ) =  \sqrt{ -2 \log( e^{-k} c(k) )  } - \sqrt{- 2\log c(k) } + O \left( \frac{ \log( - \log c(k) )}{\sqrt{ - \log c(k)}} \right).
$$ 
If $e^{-k}p(k) \downarrow 0$  as $k \downarrow - \infty$ then
$$
\IBS(k,c(k) ) =  \sqrt{ - 2 \log p(k)   }- \sqrt{ - 2 \log(e^{-k} p(k)) }  + 
O \left( \frac{ \log( - \log (e^{-k} p(k)) )}{\sqrt{ - \log(e^{-k} p(k))}} \right).
$$
where   $c(k) = 1 - e^k + p(k)$.
\end{theorem}

As before, the proof will rely on appropriate uniform bounds:

\begin{proposition} \label{th:gul-bounds} Fix $k \in \RR$ and $(1-e^k)^+ \le c < 1$.
If $k \ge 0$ we have
$$
\Phi^{-1}(c) + \sqrt{ [ \Phi^{-1}(c)]^2 + 2k } \le  \IBS(k,c )  \le \Phi^{-1}( 2 c) - \Phi^{-1}( e^{-k} c)
$$
and for $k < 0$ we have
$$
\Phi^{-1}(e^{-k}p) + \sqrt{ [ \Phi^{-1}(e^{-k}p)]^2  -  2k } \le   \IBS(k,c )  \le 
\Phi^{-1}( 2 e^{-k} p) - \Phi^{-1}( p)
$$
where $p = c + e^k - 1$.
\end{proposition}

\begin{proof}  Consider the case $k \ge 0$.  
For the upper bound,   let $d_2 = \Phi^{-1}( e^{-k} c )$  
 in Theorem \ref{th:main}.  
 
For the lower bound, let $y = \IBS(k,c)$.
 Observe  that
\begin{align*}
 \Phi\left(-\frac{k}{y}+ \frac{y}{2} \right)   &= c + e^{k} \Phi\left(-\frac{k}{y}- \frac{y}{2} \right)  \\
& \ge c.
\end{align*}
The conclusion follows from noting that the strictly increasing map
$$
x \mapsto x + \sqrt{x^2 + 2k }
$$
from $\RR$ to $(0,\infty)$ is the inverse of the map
$$
y \mapsto  - \frac{k}{y} + \frac{y}{2} 
$$
from $(0,\infty)$ to $\RR$.  

The case where $k<0$ is handled by put-call symmetry as always.
\end{proof}
 
\begin{remark}  The idea behind the lower bound is the simple inequality
$$
(X-K)^+ \le X \1_{[K, \infty)}(X)
$$
 which holds for all $X, K \ge 0$.
 \end{remark}

\begin{proof}[Proof of Theorem \ref{th:gul}]  For $k \ge 0$, we apply Proposition \ref{th:gul-bounds} and Lemma \ref{th:asym} to get
\begin{align*}
  \IBS(k, c(k) ) &\le    \Phi^{-1}( 2 c(k)) - \Phi^{-1}( e^{-k} c(k)) \\
& = - \sqrt{ -2 \log c(k)}  + \sqrt{- 2\log e^{-k} c(k) } + O \left( \frac{ \log( - \log c(k) )}{\sqrt{ - \log c(k)}} \right)
\end{align*}
where we have used $\sqrt{ -2 \log (2c)} =  \sqrt{ -2 \log c} + O( \frac{1}{\sqrt{-\log c}} )$
as $c \downarrow 0$ to
control the error from the first term, and the
bound $e^{-k} c(k) \le c(k)$ to control the error from the second term.   
Similarly, for the upper bound Proposition \ref{th:gul-bounds} and Lemma \ref{th:asym} yield
\begin{align*}
  \IBS(k, c(k) )  &\ge \Phi^{-1}(c) + \sqrt{ [ \Phi^{-1}(c)]^2 + 2k } \\
&  = - \sqrt{ -2 \log c(k)}  + \sqrt{- 2\log e^{-k} c(k) } + O \left( \frac{ \log( - \log c(k) )}{\sqrt{ - \log c(k)}} \right).
\end{align*}

The $k \downarrow - \infty$ case is similar.
\end{proof}

Figure \ref{fi:G} illustrates the behaviour of $\IBS(k,c(k))$ when $c(k) \downarrow 0$ as $k \uparrow \infty$, compared
with the uniform upper and lower bounds of Proposition \ref{th:gul-bounds} and the asymptotic
formula in Theorem \ref{th:gul}.   We have chosen the function $c(\cdot)$ 
according to the variance gamma model.  That is, we fix a time horizon $T > 0$ and let
$$
c(k ) = \EE[ (X-e^k)^+ ]
$$
where
$$
X = e^{\sigma W(G_T) + \Theta G_T + mT}
$$
and $\sigma$ and $\Theta$ are real constants, and the process $W$ is a Brownian motion
subordinated to the gamma process $G$, which is an independent L\'evy process with jump measure 
$ \mu(dx) = \frac{1}{\nu x } e^{- x/\nu} dx$ for some constant $\nu > 0$.
The constant $m$ is chosen  so that
$$
\EE[ X ] = 1.
$$
  It is well known that $G_T$ has the gamma distribution 
with mean $T$ and variance $\nu T$.  
By a routine calculation involving the moment generating functions
of the normal and gamma distribution, we find the moment generating function $M $ of $\log X $  
to be
$$
M(r) =  e^{rmT}(1 - \nu( \Theta r + \sigma^2 r^2/2) )^{-T/\nu}.
$$
Therefore, we must set
$$
m = \frac{1}{\nu} \log (1 - \nu (\Theta + \sigma^2/2) ).
$$
Note that we must assume the parameters are such that
$$
 \Theta + \sigma^2/2 < 1/ \nu
$$
to ensure that $m$ is real.  Recall that the random variable $X$ 
has the interpretation of the ratio $X = S_T/F_{0,T}$ of the time-$T$ price $S_T$ of some asset
to its initial time-$T$ forward price.  The expected value is computed under
a fixed time-$T$ forward measure.  Hence 
$c(k)$ models initial normalised price of a call option with log-moneyness
$k$ and maturity $T$.  We use  the parameters
$\sigma = 0.1213$, $\nu = 0.1686$  and $\Theta = -0.1436$ as suggested
by the calibration of Madan, Carr and Chang \cite{mcc} and set $T=5$.
 
As before, we plotted four
functions:
\begin{enumerate}
\item
 $Y_{\mathrm{upper}}(k) =  \Phi^{-1}( 2 c(k)) - \Phi^{-1}( e^{-k} c(k)) $ is the upper bound from Proposition \ref{th:gul-bounds};
 \item
 $Y_*(k) = \IBS(k,c(k))$ is the true function of our interest; 
\item 
$Y_{\mathrm{lower}}(k) =  \Phi^{-1}(c(k)) + \sqrt{ [ \Phi^{-1}(c(k))]^2 + 2k }   $ is the lower bound from Proposition \ref{th:gul-bounds};
\item 
   $Y_{\mathrm{asym}}(k) =  \sqrt{ -2 \log( e^{-k} c(k) )  } - \sqrt{- 2\log c(k) }$ is the asymptotic shape from Theorem \ref{th:gul}.
	\end{enumerate}
As always, note that $Y_{\mathrm{upper}} \ge Y_* \ge Y_{\mathrm{lower}}$ as predicted.   
	Finally,  note that $Y_{\mathrm{asym}} \le Y_{\mathrm{lower}}$ for this example.  
	It is worth remarking that for the points on the extreme right side of the graph of $Y_*$
	the moneyness $K/F_{0,T} \approx 10$ and normalised call price $C_{0,T,K}/( F_{0,T} B_{0,T}) \approx 10^{-15}$
	are outside the range of typical liquid market prices.

\begin{figure}
\caption{Bounds and asymptotics of $\IBS(\cdot, c(\cdot))$ as $c(k) \downarrow 0$ as $k \uparrow \infty$. }
	  \label{fi:G}
		\includegraphics[trim = 0cm 0 0cm 0, clip, scale = 0.38]{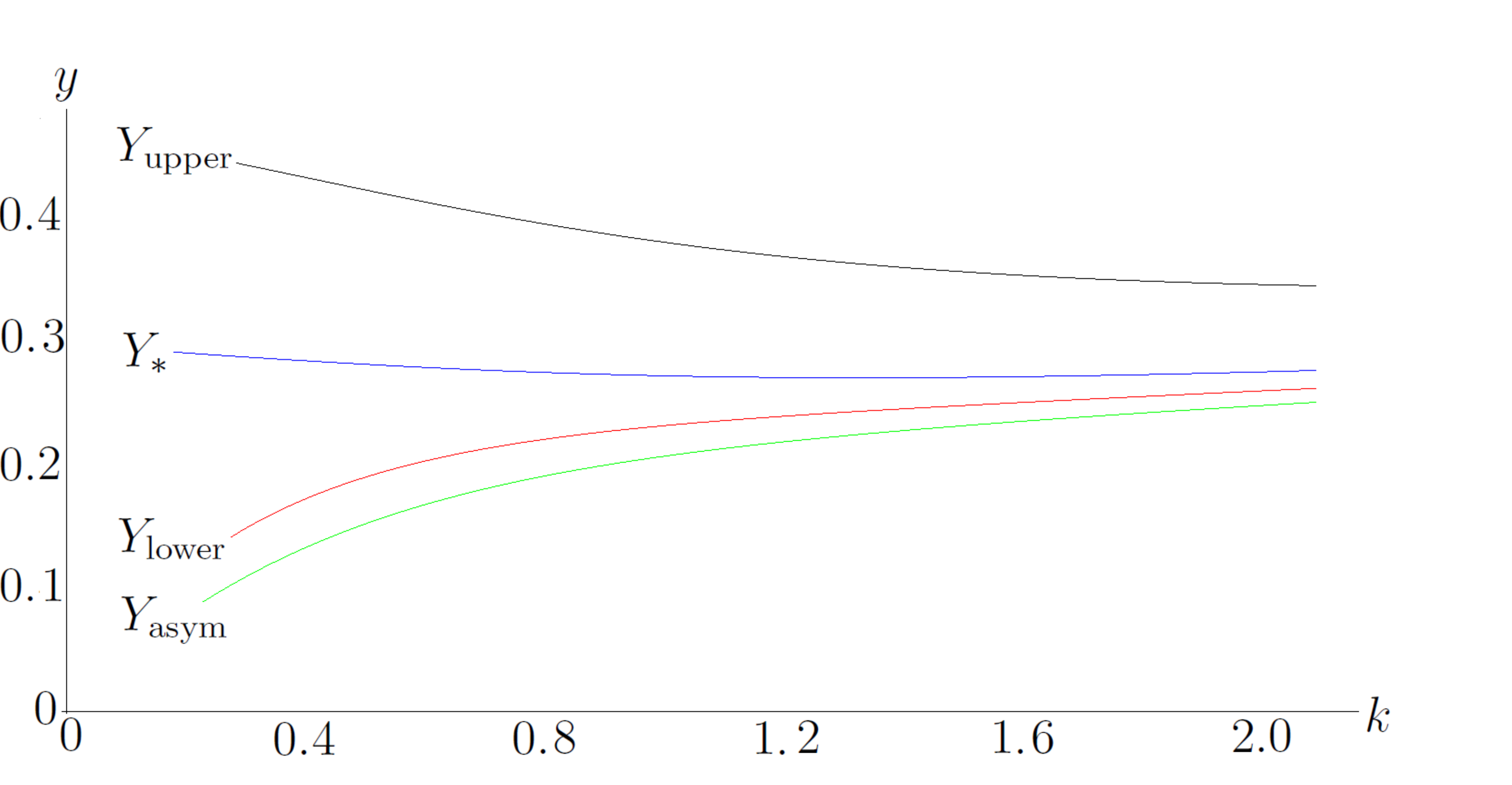}		
	\end{figure}

The recent paper \cite{DHJ} of  De Marco,   Hillairet \& Jacquier  studies a 
similar asymptotic regime as the $k \downarrow -\infty$ case of Theorem \ref{th:gul}, except 
now the assumption is that $e^{-k}p(k) \to u > 0$.   See also
the paper of  Gulisashvili \cite{gulisashvili3} for further refinements.   The motivation
is to study the left-wing behaviour of the implied volatility smile
in the case where the price of the underlying stock may hit zero. The first two terms in the 
following expansion have been known for a few years; see for instance \cite{Stockholm}.   

Even more recently, Jacquier \& Keller-Ressel \cite{JacquierKeller-Ressel}
 have interpreted the corresponding (via Proposition \ref{th:sym})
right-wing formula in terms of a market model with a price bubble.  We will comment on this
interpretation below.

\begin{theorem}[De Marco,   Hillairet \& Jacquier] \label{th:atom}
Suppose $e^{-k}p(k) \to u$ as $k \downarrow - \infty$ where $0 < u < 1$.  Then letting
$c(k) = p(k) + 1 - e^k$ we have
\begin{align*} 
\IBS(k,c(k) ) =  \sqrt{-2k} +  \Phi^{-1}(u ) + O\left ( \frac{1}{\sqrt{-k}} + \varepsilon(k) \right)
\end{align*}
 as  $k \downarrow - \infty$, where
$$
\varepsilon(k) = e^{-k}p(k) - u.
$$

\emph{(Jacquier \& Keller-Ressel)}.
Furthermore, suppose $c(k) \to u$ as $k \uparrow + \infty$ where $0 < u < 1$. Then
\begin{align*} 
\IBS(k,c(k) ) =  \sqrt{2k}  + \Phi^{-1}(u ) + O\left ( \frac{1}{\sqrt{ k}} + \varepsilon(k) \right)
\end{align*}
where 
$$
\varepsilon(k) = c(k) -u 
$$
\end{theorem}

Our proof of Theorem \ref{th:atom} reuses the uniform lower bound from Proposition \ref{th:gul-bounds}.
However, another upper bound is needed in this situation:

\begin{proposition}\label{th:lee} Fix $k \in \RR$ and $(1-e^k)^+ \le c < 1$.  If $k \ge 0$, we have 
$$
 \IBS(k,c)  \le \Phi^{-1}( c + e^k \Phi(- \sqrt{2k}) ) + \sqrt{2k} 
$$
and if $k < 0$ we have
$$
 \IBS(k,c)  \le  \Phi^{-1}(  e^{-k} p  + e^{-k} \Phi(- \sqrt{-2k}) ) + \sqrt{-2k} 
$$
where $p = c + e^k - 1$. 
\end{proposition}

\begin{proof}  In the statement of 
Theorem \ref{th:main}, let $d_2 =  -\sqrt{2k}$ if $k \ge 0$, or let  $d_1=\sqrt{-2 k}$ if $k < 0$. 
\end{proof}

\begin{proof} [Proof of Theorem \ref{th:atom}]  It is sufficient to prove only the $k<0$ case.
Recall the standard bound on the normal Mills ratio  
$$
e^{x} \Phi(-\sqrt{2 x} ) \le \frac{1}{ \sqrt{4\pi x} } \to 0  \mbox{ as } x \uparrow \infty.
$$
Hence by  Proposition \ref{th:lee} we have
\begin{align*}
 \IBS(k,c)  - \sqrt{-2k} & \le   \Phi^{-1}( u + \varepsilon(k)   + (-4\pi k)^{-1/2} ) \\
& = \Phi^{-1}(u) +  O(\varepsilon(k)   + (-k)^{-1/2} ).
\end{align*}
Similarly  by Proposition \ref{th:gul} we have
\begin{align*}
 \IBS(k,c(k) ) - \sqrt{-2k}  &\ge \Phi^{-1}(u + \varepsilon(k)) + \sqrt{[\Phi^{-1}(u + \varepsilon(k))]^2 - 2k} - \sqrt{-2k} \\
&= \Phi^{-1}(u) +  O(\varepsilon(k)   + (-k)^{-1/2} )
\end{align*}
 completing the proof.
\end{proof}

Figure \ref{fi:left} illustrates the behaviour of $\IBS(k,c(k))$ when $e^{-k}p(k) \to u > 0$ as $k \downarrow -\infty$,
where $p(k)-c(k) = e^k - 1$,  compared
with the uniform upper  of Proposition \ref{th:lee}, lower bounds of Proposition \ref{th:gul-bounds} and the asymptotic
formula in Theorem \ref{th:atom}.   We have chosen the function $c(\cdot)$ 
according to the Black--Scholes model with a jump to default.  That is, we fix
a horizon $T> 0$ and let
$$
c(k ) = \EE[ (X-e^k)^+ ]
$$
where
$$
X = \mathbf{1}_{\{T < \tau\}}e^{\sigma W_T + (\lambda- \sigma^2/2)T}
$$
and $\sigma$ and $\lambda$ are positive constants, the process $W$ is a Brownian motion
and the random variable $\tau$ is independent of $W$ and exponentially distributed 
with rate $\lambda$, so that
$$
\EE[ X ] = 1.
$$
Note that
\begin{align*}
e^{-k} p(k) & = \EE[ (1- e^{-k} X )^+ ] \\
& \to \PP(X = 0 ) \\
& = \PP( \tau \le T ) \\
& = 1 - e^{- \lambda T}.
\end{align*}
On the other hand,  it is straightforward to calculate 
$$
c(k ) = \BS( k-\lambda T, \sigma \sqrt{T} ).
$$
We use  the parameters $\sigma = 0.60$ and $\lambda=0.05$
with time horizon $T = 4$.
 
As before, we plotted four
functions: 
\begin{enumerate}
\item
 $Y_{\mathrm{upper}}(k) =  \Phi^{-1}[e^{-k}p(k) + e^{-k}\Phi(-\sqrt{-2k})] + \sqrt{-2k}$
the upper bound from Proposition \ref{th:lee};
\item  
  $Y_*(k) = \IBS(k,c(k))$ is the true function of our interest;
	\item 
	$Y_{\mathrm{lower}}(k) =  \Phi^{-1}(e^{-k} p(k)) + \sqrt{   \Phi^{-1}(e^{-k}p(k))^2 -2k }   $ is the lower bound from Proposition \ref{th:gul-bounds}; 
	\item
	 $Y_{\mathrm{asym}}(c) =  \sqrt{-2k} + \Phi^{-1}(u)$   is the asymptotic shape from Theorem \ref{th:atom}.
	\end{enumerate}
As always, note that $Y_{\mathrm{upper}} \ge Y_* \ge Y_{\mathrm{lower}}$ as predicted,
that $ Y_{\mathrm{upper}}$ is a surprisingly good approximation for $Y_*$, and that
 $Y_{\mathrm{asym}} \le Y_{\mathrm{lower}}$ for this example.  For the left-hand points of the 
graph, the moneyness $K/F_{0,T} \approx 0.04$ is somewhat outside the range of typical liquid market prices.

\begin{figure}
\caption{Bounds and asymptotics of $\IBS( \cdot , c(\cdot))$ as $e^{-k}p(k) \downarrow u > 0$ as $k \downarrow -\infty$. }
	  \label{fi:left}
		\includegraphics[trim = 0cm 0 0cm 0, clip, scale = 0.38]{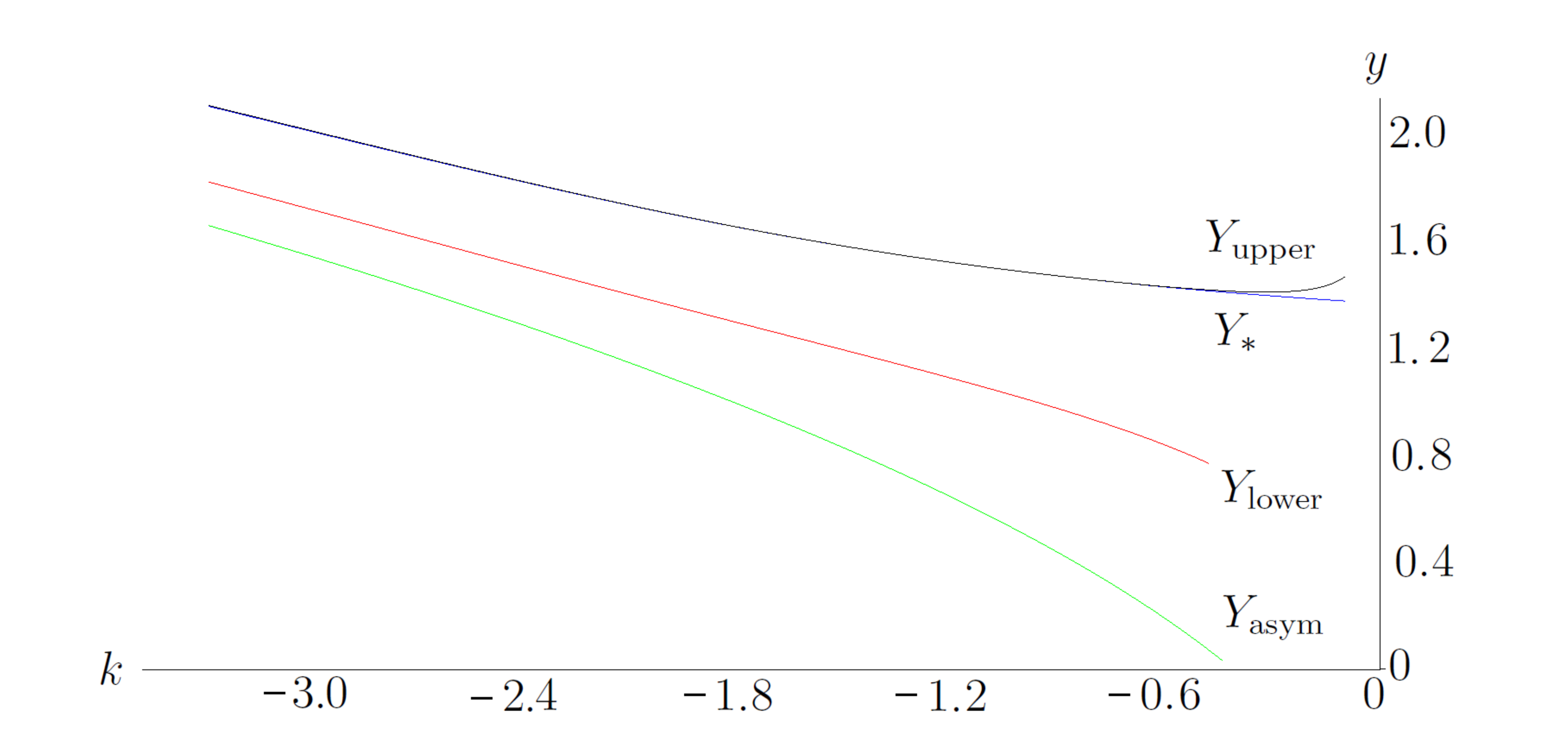}		
	\end{figure} 
	
	\begin{remark}\label{re:bubble}
To compute the implied volatility for a given model, one generally needs three ingredients:
the bond price $B_{0,T}$, the forward price $F_{0,T}$ and the call price $C_{0,T, K}$.
Consider the case where the interest rate is zero and the underlying stock pays no dividends.	
In particular, for this discussion $B_{0,T} = 1$.

In the discrete time case,   one  typically models the stock price $(S_t)_{t \ge 0}$ as  a martingale
so that there is no arbitrage.  The call price is then calculated as
$$
C_{0,T,K} = \EE[ (S_T-K)^+ ],
$$
with the justification that the above price is consistent with no-arbitrage in general, and
in the case of a complete market, the expected payout under the unique risk-neutral 
measure is the replication cost of the option and hence the unique no-arbitrage price. Similarly,
we have for the forward price the following formula
$$
F_{0,T} = \EE [ S_T ] = S_0.
$$

In the continuous time setting, things are more subtle because of the existence of doubling strategies.  If one assumes the NFLVR notion of no-arbitrage,
then by Delbaen \& Schachermayer's fundamental theorem of asset pricing \cite{DelSch}
the asset prices are sigma-martingales, but not necessarily 
true martingales. In particular, given a
dynamic model of the underlying process $(S_t)_{t \ge 0}$, this no-arbitrage
condition alone does not uniquely specify the call and forward prices, even in
a complete market.  	See, for instance, 
the paper of Ruf \cite{Ruf} for a discussion of this issue.
When the market is complete, a candidate call price is
 the minimal replication cost 
$$
C^{\mathrm{repl}} = \EE[ (S_T-K)^+ ].
$$
Another sensible way to price the call is to assume that the put price
is its minimal replication cost and the call is priced by put-call parity:
\begin{align*}
C^{\mathrm{parity}} &= S_0 - K + \EE[ (K -S_T)^+ ] \\
& = S_0 -\EE[ K \wedge S_T ]
\end{align*}
Similarly, the forward price can be either given by static replication
$$
F^{\mathrm{static}} = S_0
$$
or by dynamic replication
$$
F^{\mathrm{dyn}} = \EE[ S_T].
$$

Of course, if $(S_t)_{t \ge 0}$ is a true martingale the corresponding
candidate prices agree;
however, there has been recent interest in models where, for instance, the process
$(S_t)_{t \ge 0}$ is a non-negative strict  local martingale, and hence a strict
supermartingale.  (Such price processes are often described as bubbles;
see, for instance, the paper of Cox \& Hobson \cite{CoxHobson}.)

The result of Jacquier \& Keller-Ressel quoted here as the second
half of Theorem \ref{th:atom} corresponds to choosing
$
C_{0,T,K} = C^{\mathrm{parity}}
$
and
$
F_{0,T} = F^{\mathrm{static}}
$
so that the implied volatility is 
$$
\sigma^{\textrm{implied}} =  \frac{1}{\sqrt{T}}
Y( \log(K/S_0),   1 - \EE[ K \wedge S_T ]/S_0).
$$
We note here that this 
  convention for defining implied volatility was also adopted
in \cite{JAP09}.  On the other hand, note that
the convention $C_{0,T,K} = C^{\mathrm{repl}}$ and 
$F_{0,T} = F^{\mathrm{dyn}}$ is used in equation \eqref{eq:intro} of the introduction.
	\end{remark}

We conclude with some remarks on the bounds and asymptotic formulae 
in this section.  The numerical results suggest that 
for at least some situations, one of the upper or lower bounds 
is a better approximation to the implied total standard deviation than
the corresponding lowest order asymptotic formula.  One could argue that 
with more terms in the asymptotic series, better accuracy could be attained
with the asymptotics.
Although such a claim is indeed plausible, there are a few reasons why it is beside the point.

First, the numerical results presented here should only be considered a proof of concept, rather
than a head-to-head competition between state-of-the-art approximations.
Nevertheless,  it is worth noting both the given bounds and the asymptotic formulae are only approximations, and therefore
have an error term.  But  unlike the error terms of an asymptotic formula, the error term for our bounds have
a known sign.

Second, given one bound,
the theorems of Section \ref{se:main} give a systematic way of finding a better
bound.   Indeed, fix $(k,c)$ with $k > 0$, and let $y_* = \IBS(k,c)$.  Suppose it known that
$$
y_* < y_1
$$
where $y_1$ is some given approximation.  Define $F:(y_{\mathrm{min}} ,\infty) \to (0, \infty)$ by 
$$
F(y) = H_1( -k/y + y/2; k, c)
$$
where 
$$
y_{\mathrm{min}} = \Phi^{-1}(c) + \sqrt{ [ \Phi^{-1}(c)]^2 + 2k }
$$
and $H_1$ is the functions
defined by equation \eqref{eq:H1} of Section \ref{se:main}.
Letting
$$
y_2 = F(y_1)
$$
  we have by Theorem \ref{th:main} that
$$
y_* < y_2.
$$
However, more is true.   Note that the map $F$ has 
a unique fixed point $y_*$.  Since
	$$
	\lim_{y \downarrow y_{\mathrm{min}} } F(y) = \infty
	$$
	we conclude by the continuity of $F$ that $F(y) > y$ for $ y_{\mathrm{min}} < y <  y^*$,
	and more importantly, that $F(y) < y$ for $y > y^*$.   
	In particular, 
	$$
	y_2 = F(y_1) < y_1.
	$$
	That is, $y_2$ is a \textit{better} approximation of $y_*$ and the error term has the same
	sign as the original approximation.  Of course, this process can iterated. Letting 
$y_n = F(y_{n-1})$ we see that the sequence $(y_n)_{n \ge 1}$
is decreasing  and $\inf_n y_n = y_*$.  

Notice that this sequence converges very rapidly. Indeed, by  Taylor's theorem 
\begin{align*}
y_n &= F(y_{n-1}) \\
& = F(y^*) + F'(y_*) ( y_{n-1}- y_*) + \frac{1}{2}F''(\hat{y})( y_{n-1}- y_*)^2  
\end{align*}
for some $y_* < \hat y < y_{n-1}$.
Since $y_*$ minimises $F$ we have 
$$
F'(y_*) = 0
$$
and hence, by the continuity of $F''$, we have
$$
	\frac{ y_n -y_*}{(y_{n-1}- y_*)^2} \to \frac{1}{2} F''(y^*) = \frac{1}{2 y_*} \left( \frac{k}{y_*}+ \frac{y_*}{2} \right)^2
	$$
	as $n \to \infty$.

\begin{figure}	
\caption{The cobweb diagram illustrating the convergence of $y_0, y_1, y_2, \ldots$ to the fixed point $y_* = F(y_*)$.  }
	  \label{fi:cobweb}
		\includegraphics[trim = 0cm 0 0cm 0, clip, scale = 0.4]{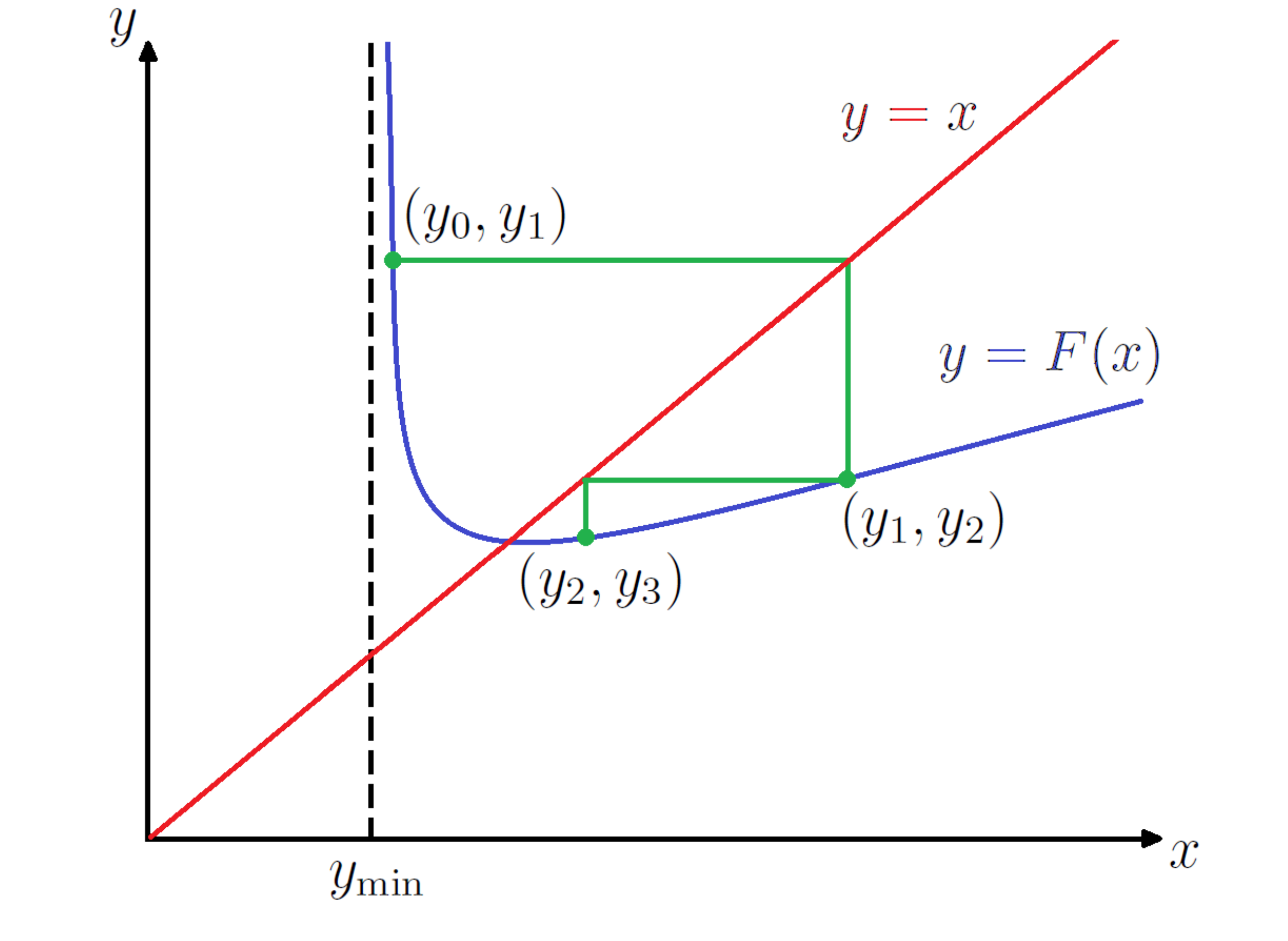}		
	\end{figure}
	
Furthermore, we can find our initial upper bound $y_1$ by 
choosing any $y_0 > y_{\mathrm{min}}$ and letting $y_1 = F(y_0)$.  This
procedure is  illustrated by the cobweb
diagram of Figure \ref{fi:cobweb}.   Of course, the convergence can be helped along
by an inspired choice of $y_0$ as discussed at the beginning of this section.

The above discussion of a rapidly converging sequence 
should be contrasted with the approach taken, for instance, in 
the paper of Gao \& Lee \cite{GaoLee}.  There a systematic method of computing 
terms in the asymptotic series for implied volatility is obtained.  However,
unlike the procedure discussed above, an asymptotic
series  may diverge as more terms 
are added.  

A third and final point is that the approximations for the implied total standard deviation 
are not particularly interesting on their own. Indeed, to use the formulae in
Proposition \ref{th:gul-bounds} one must already know the normalised bond price $c(k)$.
If this quantity is to be calculated numerically from a certain model, one might 
as well compute the $\IBS(k, c(k))$ numerically also.  The point of these bounds is to 
be used in conjunction with other,  model dependent bounds on $c(k)$ to obtain useful
bounds on the quantities of interest.

\section{Acknowledgement} This work was presented at the
Conference on Stochastic Analysis for Risk Modeling in Luminy and the
Tenth Cambridge--Princeton Conference in Cambridge.  I would like to
thank the participants for their comments. I would also like to thank
the anonymous referees whose comments on   earlier drafts of this paper
greatly improved the present presentation. 
Finally I would like to thank the Cambridge Endowment for Research in Finance for 
their support.

\end{document}